\documentclass[letterpaper, 10 pt, conference]{ieeeconf}
\IEEEoverridecommandlockouts
\usepackage{cite}

\usepackage{amsmath,amssymb,amsfonts,amsthm,mathptmx}
\usepackage{algorithmic}
\usepackage{graphicx}
\usepackage[font=small]{caption}

\makeatletter               % bring back color and hyperlinks to citations
\let\NAT@parse\undefined
\makeatother
\usepackage{hyperref}
\hypersetup{colorlinks,breaklinks,
linkcolor=[rgb]{0.5,0.,0.},
citecolor=[rgb]{0.000,0.427,0.173},
urlcolor=[rgb]{0.031,0.318,0.612}}

\newtheorem{thm}{Theorem}[section]
\newtheorem{prop}[thm]{Proposition}

\usepackage{textcomp}
\usepackage{xcolor}
\def\BibTeX{{\rm B\kern-.05em{\sc i\kern-.025em b}\kern-.08em
    T\kern-.1667em\lower.7ex\hbox{E}\kern-.125emX}}

\title{\LARGE \bf Free Energy Principle for the Noise Smoothness Estimation of \\  Linear Systems with Colored Noise  \thanks{Both authors are with the Department of Cognitive Robotics at TU Delft, The Netherlands. Corresponding author: ajitham1994@gmail.com} }

\author{Ajith Anil Meera and Martijn Wisse}

\begin{document}

\maketitle
\thispagestyle{empty}
\pagestyle{empty}

\begin{abstract}
The free energy principle (FEP) from neuroscience provides a framework called active inference for the joint estimation and control of state space systems, subjected to colored noise. However, the active inference community has been challenged with the critical task of manually tuning the noise smoothness parameter. To solve this problem, we introduce a novel online noise smoothness estimator based on the idea of free energy principle. We mathematically show that our estimator can converge to the free energy optimum during smoothness estimation. Using this formulation, we introduce a joint state and noise smoothness observer design called DEMs. Through rigorous simulations, we show that DEMs outperforms state-of-the-art state observers with least state estimation error. Finally, we provide a proof of concept for DEMs by applying it on a real life robotics problem - state estimation of a quadrotor hovering in wind, demonstrating its practical use.

\end{abstract}

% \begin{IEEEkeywords}
% Active inference, noise estimation, free energy principle, colored noise.
% \end{IEEEkeywords}

\section{INTRODUCTION}
The rising demand for autonomous drone delivery systems has increased the need for accurate state observers that are robust against uncertain events like strong wind currents. These unmodelled wind currents induce colored noise to the system, hindering the safe operation of drones. We take a step in this direction by using the ideas from computational neuroscience to introduce a novel state and noise smoothness observer design for linear systems with colored noise.

% The accurate state estimation under colored noise is a practically relevant problem in robotics. For example, the unmodelled wind dynamics acting on a delivery drone in flight induces colored noise in the system. Since the classical estimators like Kalman Filter (KF) provides sub-optimal solutions under non-white noises, the quality of state estimation and control deteriorates. Therefore, online noise smoothness estimation is important for the safe operation of these robots in real environments. To solve this problem, we propose a novel online noise smoothness estimator, inspired by the free energy principle from neuroscience. 

The classical linear estimators like Kalman Filter (KF) assumes the noises to be white. This assumption is often violated in practice, resulting in a sub-optimal estimation \cite{crassidis2004optimal}. Many adaptations on KF have been introduced to overcome this challenge, including Second Moment Information Kalman filter (SMIKF) \cite{zhouSMIKF}, State Augmentation (SA) \cite{brysonSA}, Measurement Differencing (MD) \cite{brysonMD} etc.  Dynamic Expectation Maximization (DEM) \cite{friston2008variational}, based on Free Energy Principle (FEP) \cite{friston2010free} from neuroscience has recently been used to design state and input observers \cite{meera2020free} that has shown to outperform the classical methods both in simulation and in real robot experiments \cite{bos2021free}. However, DEM requires the prior knowledge of the noise smoothness parameter for the accurate state estimation. To solve this problem, we introduce a novel online noise smoothness estimator based on FEP, for linear systems with colored noise. The core contributions of the paper include: 
\begin{enumerate}
    \item introduction of an online smoothness estimator for the state estimation of linear systems under colored noise,
    \item extensive evaluation of the estimator in simulation and its validation on a real robot (quadrotor flight) data.
\end{enumerate}

\section{RELATED WORK}
This section highlights the interdisciplinary nature of FEP with related works in neuroscience and robotics literature.
\subsection{Neuroscience}
FEP posits that all biological systems resist their natural tendency to disorder by minimizing their free energy \cite{friston2010free}, where free energy is an information theoretic measure that bounds sensory surprisal. FEP emerges as a unified theory of the brain by providing a mathematical formalism for brain functions \cite{friston2009free}, unifying action and perception \cite{friston2010action}, explaining Freudian ideas \cite{carhart2010default}, and connecting memory and attention \cite{friston2009free}.  The work closest to our proposed idea is the Generalized filtering \cite{friston2010generalised} that uses FEP for noise smoothness estimation during the inversion of dynamic models of the brain (fMRI data) \cite{li2011generalised}. We extend this idea into robotics to design an online state and noise smoothness observer for applications like quadrotor flights with wind as colored noise.

\subsection{Robotics and control}
In control systems literature, numerous approaches are used to deal with colored noise during state estimation \cite{crassidis2004optimal}. SA models the colored process noise as an auto-regressive (AR) noise and augments the state space equation to transform it into an equivalent system influenced by white noise \cite{brysonSA}. SMIKF \cite{zhouSMIKF} extends KF for coloured noise by incorporating the temporal correlations of the AR noise into the prior covariance calculation of KF. MD \cite{brysonMD} approach deals with handling colored measurement noise. However, the white noise assumption is prevalent in robotics for the state estimation of a quadrotor \cite{xiong2015optimal}, which might not be effective in outdoor windy conditions \cite{bos2021free}. Our work fills this research gap by providing an online noise smoothness estimator.

The brain inspired nature of FEP has already inspired the development of intelligent agents \cite{lanillos2021active} -- body perception of humanoid robots \cite{oliver2019active}, estimation and control of manipulator robot \cite{baioumy2021active}, system identification of a quadrotor \cite{meera2021brain}, SLAM \cite{ccatal2021robot}, PID controller \cite{baltieri2019pid}, KF \cite{bos2021free,baltieri2021kalman} etc. These active inference applications can employ our noise estimator for better estimation and control of robots during colored noise.

\section{PROBLEM STATEMENT}
Consider the linear plant dynamics given in Equation \ref{eqn:general_LTI} where $A$, $B$ and $C$ are constant system matrices, $\textbf{x}\in \mathbb{R}^n$ is the hidden state, $\textbf{v}\in \mathbb{R}^r$ is the input and $\textbf{y}\in \mathbb{R}^m$ is the output.
\begin{equation} 
\label{eqn:general_LTI}
    \begin{split}
       \Dot{\textbf{x}} = A\textbf{x}+B\textbf{v} + \textbf{w},
    \end{split}
    \quad \quad
    \begin{split}
        \textbf{y} = C\textbf{x} +\textbf{z}.
    \end{split}{}
\end{equation}{}Here $\textbf{w}\in \mathbb{R}^n$ and $\textbf{z}\in \mathbb{R}^m$ represent the process and measurement noise with noise precision (inverse covariance) $\Pi^w$ and $\Pi^z$  respectively. Variables of the plant are denoted in boldface, while its estimates are denoted in non-boldface. The noises in this paper are generated through the convolution of white noise with a Gaussian filter of kernel width $s$. 

The problem considered in this paper is the state ($x$) and noise smoothness ($s$) observer design (DEMs) for the linear system given in Equation \ref{eqn:general_LTI}, subjected to colored noise. We show that our observer outperforms state-of-the-art state observers, both in simulation (Section \ref{sec:benchmarking}) and on real robot data (Section \ref{sec:quadrotor}).

\section{NOISE COLOR MODELLING}
The two key concepts behind the success of DEM in handling the colored noise are i) the use of generalized coordinates and ii) the noise precision modelling. This section aims to elaborate on these theoretical concepts.  

\subsection{Generalized coordinates}
Generalized coordinates is a vector representation of the trajectory of a time varying quantity ($x, v, y$) using a collection of its higher order derivatives. For example, the state vector in generalized coordinates is written using a tilde operator as $\tilde{x}  = [x \ x' \ x'' \ ....]^T$, where the dash operator represents the derivatives. The key advantage is the capability to track the trajectory of states, unlike the classical estimators that track only the point estimates. This provides additional data for DEM during estimation, resulting in its superior performance during state estimation under colored noise. Since the noises are colored, the higher derivatives of the system model can be written as \cite{friston2008variational}: 
\begin{equation} 
    \begin{split}
        &x'=Ax+Bv+w \\
        &x''=Ax'+Bv'+w'\\ &...
    \end{split}
    \quad \quad
     \begin{split}
        &y=Cx+z \\
        & y' = Cx'+z'\\ &...
    \end{split}   
\end{equation}{} 
which can be compactly written as: \begin{equation}
\label{eqn:generative_process}
    \begin{split}
        &\dot{\tilde{{x}}}  = D^x\tilde{{x}}= \tilde{A}\tilde{{x}}+\tilde{B}\tilde{{v}} +\tilde{{w}} 
    \end{split}{}
    \quad \quad
    \begin{split}
        &\tilde{{y}} = \tilde{C}\tilde{{x}}+\tilde{{z}}
    \end{split}{}
\end{equation}{} where $D^x = \Bigg[ \begin{smallmatrix}{}
0 & 1 & & &\\
 & 0 & 1 & & \\
 & & .& . &  \\
 & & & 0& 1 \\
 & & & & 0
\end{smallmatrix} \Bigg]_{(p+1)\times (p+1)} \otimes I_{n\times n}.$\\ 
Here, $D^x$ represents the shift matrix, which performs the derivative operation on the generalized state vector. $p$ and $d$ represent the embedding order for the hidden states and the inputs respectively, indicating the number of derivatives used. The generalized system matrices are given by $ \tilde A = I_{p+1} \otimes A,\hspace{10pt} \tilde B = I_{p+1} \otimes B,\hspace{10pt}\tilde C = I_{p+1} \otimes C$, where $I$ denotes the identity matrix and $\otimes$ the Kronecker tensor product. 
% The generalized output $\tilde y$ is calculated from the discrete measurements $\hat y = \Bigg[ \begin{smallmatrix} \hdots \\ y(t-dt) \\ y(t) \\ y(t+dt) \\ \hdots \end{smallmatrix} \Bigg]_{m(p+1)}$ using the methodology in \cite{meera2020free}, resulting in a latency of $\frac{p}{2}dt$ during online estimation, which is negligible for practical robotics applications with large sampling rates like 120Hz used in Section \ref{sec:quadrotor}.

\subsection{Noise precision modelling}
The second key concept is the modelling of generalized noise precision (inverse covariance) matrix $\tilde{\Pi}$. Since the noises are assumed to be Gaussian convoluted white noise, the covariance matrix embedding the relation between noise derivatives take a specific structure \cite{friston2008variational}. The smoothness matrix defining this relation for $p=6$ is calculated as \cite{meera2020free}:
\begin{equation} \label{eqn:S_matrix}
    S = \begin{bmatrix}
    \frac{35}{16} & 0 & \frac{35}{8}s^2 & 0 & \frac{7}{4} s^4 & 0 & \frac{1}{6} s^6 \\
    0 & \frac{35}{4} s^2 & 0 & 7s^4 & 0 & s^6 & 0 \\
    \frac{35}{8}s^2 & 0 & \frac{77}{4}s^4 & 0 & \frac{19}{2} s^6 & 0 & s^8 \\
    0 & 7s^4 & 0 & 8s^6 & 0 & \frac{4}{3}s^8 & 0 \\
    \frac{7}{4} s^4 & 0 & \frac{19}{2} s^6 & 0 & \frac{17}{3}s^8 & 0 & \frac{2}{3}s^{10} \\
    0 & s^6 & 0 & \frac{4}{3} s^8 & 0 & \frac{4}{15}s^{10} & 0 \\
    \frac{1}{6} s^6 & 0 & s^8 & 0 & \frac{2}{3}s^{10} & 0 & \frac{4}{45}s^{12} \\
    \end{bmatrix},
\end{equation}
where $s$ is the kernel width of the Gaussian filter. Since $s<1$ second for practical cases (sensors have high sampling rate), the first elements in $S$ matrix are higher than the last ones, implying a higher correlation between the first noise derivatives (more smooth) than the last derivatives (less smooth). The combined (generalized) noise precision matrix can be written using the $S$ matrix as:
\begin{equation} \label{eqn:gen_precision}
    \tilde{\Pi} = \begin{bmatrix}
    S \otimes \Pi^z & 0 \\ 0 & S \otimes \Pi^w
    \end{bmatrix}.
\end{equation}
With the key concepts in place, the next section derives the free energy formulations that are necessary for the observer design in Section \ref{sec:observer_design}. 

\section{FREE ENERGY OPTIMIZATION}
FEP uses Bayesian Inference to estimate the posterior probability $p(\vartheta/\textbf{y}) = {p(\vartheta,\textbf{y}}/{\int p(\vartheta,\textbf{y})d\vartheta}$, where $\vartheta$ is the component to be estimated ($\vartheta = \{\tilde{x},s\}$), and $\textbf{y}$ is the measurement \cite{buckley2017free}. The presence of an intractable integral motivates the use of a variational density $q(\vartheta)$, called the recognition density that approximates the posterior as $q(\vartheta)\approx p(\vartheta/\textbf{y})$. This approximation is achieved by minimizing the Kullback-Leibler (KL) divergence of the distributions given by $KL(q(\vartheta)||p(\vartheta/\textbf{y})) = \langle \ln q(\vartheta)\rangle_{q(\vartheta)} - \langle \ln{p(\vartheta/\textbf{y})}\rangle_{q(\vartheta)}$, where $\langle.\rangle_{q(\vartheta)}$ represents the expectation over $q(\vartheta)$. Upon simplification using $p(\vartheta/\textbf{y}) = {p(\vartheta,\textbf{y})}/{p(\textbf{y})}$, it reduces to \cite{friston2010free}: 
\begin{equation}
    \label{eqs:logevidence}
    \ln p(y) = F + KL(q(\vartheta)||p(\vartheta|\textbf{y})),
\end{equation}
where $ F =  \langle \ln{p(\vartheta,\textbf{y})}\rangle_{q(\vartheta)}  -\langle \ln q(\vartheta)\rangle_{q(\vartheta)}$ is the free energy. Since $\ln p(\textbf{y})$ is independent of $\vartheta$, minimization of the KL divergence involves the maximization of free energy. This is the fundamental idea behind using free energy as the proxy for brain's inference through the minimization of its sensory surprisal \cite{friston2010free}. 

We use this idea from free energy principle for the joint observer design for $\tilde{x}$ and $s$ through two fundamental assumptions about $q(\vartheta) = q(\tilde{x},s)$: i) Mean field assumption \cite{friston2008variational} that facilitates a conditional independence between the subdensities, $q(\vartheta) = q(\tilde{x})q(s)$,  and ii) Laplace assumption \cite{friston2007variational} that facilitates the use of Gaussian distributions with mean $\mu$ and variance $\Sigma$ over these subdensities, $q(\tilde{x})= \mathcal{N}(\tilde{x}:\mu^{\tilde{x}},\,\Sigma^{\tilde{x}})$ and $q(s)= \mathcal{N}(s:\mu^{s},\,\Sigma^{s})$. We refer to \cite{anilmeera2021DEM_LTI} for an elaborate read on similar simplifications. Under these assumptions, $F$ reduces to the sum of precision weighted prediction errors and the information entropy as:
\begin{equation} \label{eqn:full_F_exp}
    F = -\frac{1}{2}\tilde{\epsilon}^T\tilde{\Pi}\tilde{\epsilon} + \frac{1}{2} \ln |\tilde{\Pi}| - \frac{1}{2} \epsilon^{sT} \Pi^s \epsilon^s +  \frac{1}{2} \ln |\Pi^s| ,
\end{equation}
where $\tilde{\epsilon}$ is the combined prediction error for outputs and states, and $\epsilon^s$ is the prediction error for $s$, given by:
\begin{equation} 
 \tilde{\epsilon} = \begin{bmatrix}
    \tilde{\textbf{y}} - \tilde{C} \tilde{x} \\
    D\tilde{x} - \tilde{A} \tilde{x} - \tilde{B} \tilde{v}
\end{bmatrix}, \text{ and } \epsilon^s = s-\eta^s   .
\end{equation}
Here $\eta^s$ and $\Pi^s$ are the prior smoothness and its prior precision (confidence). For this work, we start the estimation from a low prior $\eta^s \approx 0$ with a very low confidence $(\Pi^s = 1)$. Therefore, Equation \ref{eqn:full_F_exp} reduces to:
\begin{equation} \label{eqn:free_energy}
    F = -\frac{1}{2}\tilde{\epsilon}^T\tilde{\Pi}\tilde{\epsilon} + \frac{1}{2} \ln |\tilde{\Pi}| - \frac{1}{2}s^2.
\end{equation}
The last term in Equation \ref{eqn:free_energy} is the novel term that we have introduced for optimizing smoothness, and doesn't appear in FEP literature. Using this, we propose an online noise smoothness estimation algorithm which estimates $s$ through the gradient ascend (maximization) of $F$, where $\frac{\partial F}{\partial s} |_{s=s_{o}} = 0$ and $\frac{\partial^2 F}{\partial s^2} |_{s=s_{o}} < 0$, with $s_o$ being the smoothness value that maximizes $F$. The free energy gradients necessary for this scheme are obtained by differentiating Equation \ref{eqn:free_energy}:
\begin{equation} \label{eqn:F_gradients}
\begin{split}
        \frac{\partial F}{\partial s} & = -\frac{1}{2} \tilde{\epsilon}^T \frac{\partial \tilde{\Pi}}{\partial s}\tilde{\epsilon} + \frac{1}{2} \frac{\partial \ln |\tilde{\Pi}|}{\partial s} - s \\
        \frac{\partial^2 F}{\partial s^2} & = -\frac{1}{2} \tilde{\epsilon}^T \frac{\partial^2 \tilde{\Pi}}{\partial s^2}\tilde{\epsilon} + \frac{1}{2} \frac{\partial^2 \ln |\tilde{\Pi}|}{\partial s^2} - 1,
\end{split}
\end{equation}
where the gradients of $\ln |\tilde{\Pi}|$  can be computed as (refer Appendix \ref{app:gradients}) :
\begin{equation} \label{eqn:dlog_Pi}
        \frac{\partial \ln |\tilde{\Pi}|}{\partial s}  = 42(n+m) \frac{1}{s}, \
        \frac{\partial^2 \ln |\tilde{\Pi}|}{\partial s^2}  = -42(n+m) \frac{1}{s^2}. 
\end{equation}
The usage of a gradient ascent scheme on the free energy curve for the estimation of $s$ is motivated by the proof for the existence of a unique maximum for $F$ under practical bounds as follows.
\begin{prop} \label{prop:FE_max}
The free energy $F$ defined by Equation \ref{eqn:free_energy} has a unique maximum with respect to noise smoothness $s$, under the practical range of noise smoothness $(0<s<1)$ and sampling time $(dt<s)$.
\end{prop}
\begin{proof}
Consider all the smoothness values of $s$ with zero free energy gradients ($\frac{\partial F}{\partial s} |_{s=s_{o}} = 0$). Substituting Equation \ref{eqn:dlog_Pi} in \ref{eqn:F_gradients} and using $\frac{\partial F}{\partial s} |_{s=s_{o}} = 0$ yields the condition satisfied by all maximum and minimum points:
\begin{equation} \label{eqn:s_opt}
    \frac{1}{2} (\tilde{\epsilon}^T \tilde{\Pi}_s \tilde{\epsilon})|_{s=s_o} = \frac{21(n+m)}{s_o} - s_o
    % s_{o} =  \frac{42(n+m)}{(\tilde{\epsilon}^T \tilde{\Pi}_s \tilde{\epsilon})|_{s=s_o}},
\end{equation}
where we use the shorthand $\tilde{\Pi}_s = \frac{\partial \tilde{\Pi}}{\partial s}$. Since $0<s<1$, we have from Equation \ref{eqn:s_opt} that:
\begin{equation} \label{eqn:grad_condition}
        (\tilde{\epsilon}^T \tilde{\Pi}_s \tilde{\epsilon})|_{s=s_o} > 0 .
\end{equation}
The proof for the existence of a unique maximum is complete if we prove that $\frac{\partial^2 F}{\partial s^2}|_{s=s_o}<0$, for all $s_o$ satisfying Equation \ref{eqn:grad_condition}. The curvature of $F$ at $s=s_o$ is calculated from Equation \ref{eqn:F_gradients} using Equation \ref{eqn:dlog_Pi} as:
\begin{equation} \label{eqn:Fss_so}
% \begin{split}
%         F_{ss} &= -\frac{1}{2} \Big( \tilde{\epsilon}^T \tilde{\Pi}_{ss} \tilde{\epsilon}  + 42(n+m) \frac{1}{s^2} + 1 \Big), \\
        F_{ss}|_{s=s_o} = -\frac{1}{2} \Big( (\tilde{\epsilon}^T  \tilde{\Pi}_{ss} \tilde{\epsilon})|_{s=s_o}  +  42(n+m) \frac{1}{s_o^2} + 1 \Big)
% \end{split}
\end{equation}
Since $\tilde{\Pi} \succ 0$, from definition $\tilde{\epsilon}^T \tilde{\Pi} \tilde{\epsilon} > 0$, and since $(\tilde{\epsilon}^T \tilde{\Pi}_s \tilde{\epsilon})|_{s=s_o} > 0$, we can conclude that $ (\tilde{\epsilon}^T  \tilde{\Pi}_{ss} \tilde{\epsilon})|_{s=s_o} > 0$, even though $\tilde{\Pi}_s \nsucc 0$ and $\tilde{\Pi}_{ss} \nsucc 0$ (refer Appendix \ref{app:numerical} for numerical analysis). From Equation \ref{eqn:Fss_so}, $( \tilde{\epsilon}^T  \tilde{\Pi}_{ss} \tilde{\epsilon})|_{s=s_o} > 0 \implies F_{ss}|_{s=s_o} < 0$, completing the proof for the existence of a unique maximum of free energy at $s=s_o$. 
\end{proof}

\section{OBSERVER DESIGN} \label{sec:observer_design}
This section aims to introduce a novel observer design (DEMs) for the joint state and noise smoothness estimation of a linear system with colored noise. We formulate the noise smoothness estimator from the previous section (gradient ascend on $F$), using the Newton-Gauss update scheme:
\begin{equation} \label{eqn:s_update_rule}
\begin{split}
        s(t+ dt) &= s(t) + ds, \\
        ds &= (e^{F_{ss}|_t dt}-1)(F_{ss}|_t)^{-1}F_s |_t,
\end{split}
\end{equation}
where $s(t)$ is the smoothness at time $t$, and $ds$ is the smoothness increment for a time increment of $dt$. We combine this observer design with the standard DEM observer design for state estimation \cite{meera2020free}, where the update equation in the continuous time is given by:
\begin{equation} \label{eqn:DEM_state_update}
    \dot{\tilde{x}} =     A_1 \tilde{x} +  
    B_1 \begin{bmatrix}
    \tilde{\textbf{y}} \\ \tilde{\textbf{v}}
    \end{bmatrix}  
\end{equation}
where $A_1 = [D^x-k^x\tilde{C}^T\tilde{\Pi}^z\tilde{C}  -k^x(D^x-\tilde{A})^T \tilde{\Pi}^w(D^x-\tilde{A})]$, $B_1 = k^x\begin{bmatrix}
    \tilde{C}^T\tilde{\Pi}^z & (D^x-\tilde{A})^T\tilde{\Pi}^w\tilde{B}
    \end{bmatrix}$, and $k^x$ is the learning rate which is set to 1 throughout this paper. Since Equation \ref{eqn:DEM_state_update} is a linear differential equation, an exact algebraic discretization can be performed for the observer as:
\begin{equation} \label{eqn:state_update_rule}
    \tilde{x}(t+dt) = e^{A_1 dt} \tilde{x}(t) + A_1^{-1}(e^{A_1 dt} - I)B_1 \begin{bmatrix}
    \tilde{\textbf{y}}(t) \\ \tilde{\textbf{v}}(t)
    \end{bmatrix}  
\end{equation}
Equations \ref{eqn:s_update_rule} and \ref{eqn:state_update_rule} together complete our observer design. Note that $A_1$ and $B_1$ are nonlinear functions of $s$ because of the presence of $\tilde{\Pi}^w$ and $\tilde{\Pi}^z$ in it. Therefore, the update equations of state and noise smoothness observers are coupled. Since this heavily complicates the stability proof of the joint estimator, we leave it for future research.

\section{WORKING EXAMPLE}
This section aims to provide a working example in simulation to show the capabilities of our observer design. We use simulation data at different $s$ levels to show that DEMs can accurately estimate $\tilde{x}$ and $s$.

\begin{figure}[h]
    \centering
    \captionsetup{justification=justified}
    \includegraphics[scale = 0.25]{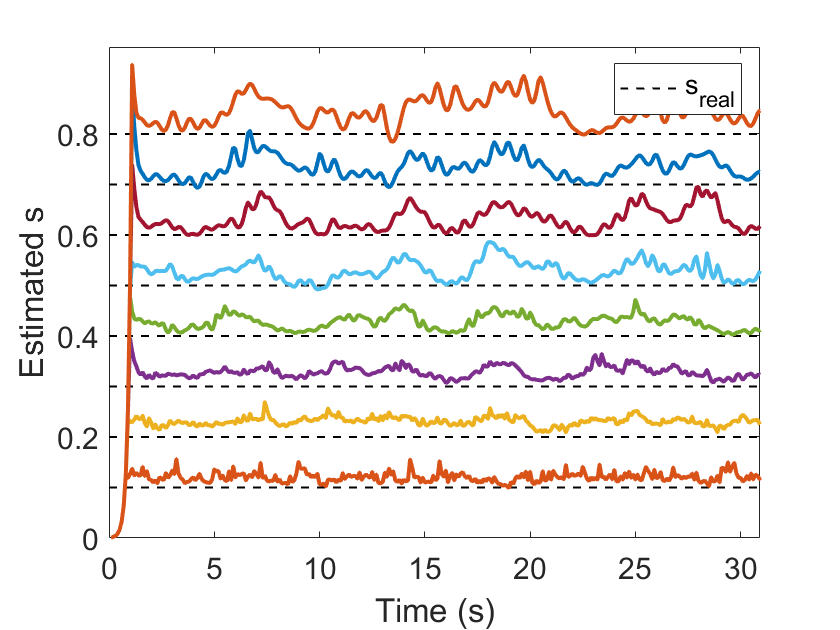}
    \caption{The maximization of $F$ successfully estimates $s$ for 8 simulations with different $s_{real}$. The colored solid lines represent the online estimation of $s$, whereas the dotted lines represent $s_{real}$. The estimation starts from  $\eta^s=0.001$ at time $t=0$ for all simulations and converges close to $s_{real}$ with a bias, within a few samples.  }
    \label{fig:s_est_time}
\end{figure}

\subsection{Simulation settings} \label{sec:sim_setup}
A random system with $A =  \begin{bmatrix}
    0.0484   & 0.7535 \\  -0.7617 &  -0.2187
    \end{bmatrix}$, $B = \begin{bmatrix}
    0.3604 \\     0.0776
    \end{bmatrix}$, and $C =  \begin{bmatrix}
    0.2265 & -0.4786\\ 0.4066 &  -0.2641\\ 0.3871 & 0.3817 \\ -0.1630 & -0.9290
    \end{bmatrix}$ was used to generate the synthetic data for a total time of $T = 32s$ with increments $dt=0.1s$, and a Gaussian bump input $v = e^{-0.25(t-12)^2}$. The colored noise was generated using $\Pi^w = e^6 I_{2}$ and $\Pi^z = e^6 I_{4}$. This simulation setting will be used throughout the paper, unless mentioned otherwise.  We generate eight such time series data using different levels of noise smoothness $s$, ranging from 0.1 to 0.8 and use it for the analysis in this section. 

\begin{figure}[h]
    \centering
    \captionsetup{justification=justified}
    \includegraphics[scale = 0.27]{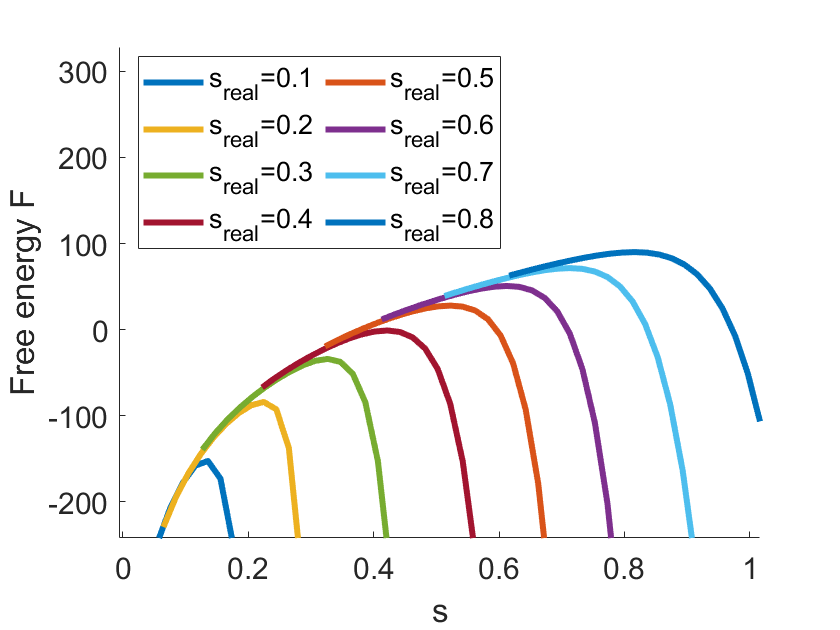}
    \caption{Free energy vs $s$ plot for 8 different simulations (with 8 different real $s$ values) from Figure \ref{fig:s_est_time} at $t=5s$. All 8 free energy curves show a clear maximum around the real $s$ values. This shows the effectiveness of our noise smoothness observer.   }
    \label{fig:F_peaks}
\end{figure}

\subsection{Test example}
Figure \ref{fig:s_est_time} shows the results of our noise smoothness estimator for all eight simulations. All simulations start with the prior $\eta^s= 0.001$ and quickly stabilises around the correct smoothness value ($s_{real}$ in dashed black), showing the success of our estimator for a range of noise smoothness values. Figure \ref{fig:F_peaks} shows the free energy vs $s$ curve at $t=5s$ for all eight simulations. The clear peaks of the free energy curve around the correct noise smoothness value ($s_{real}$) shows that free energy could be used as the objective function for noise estimation for the operational ranges of $s$. The importance of estimating the correct $s$ is shown in Figure \ref{fig:SSE_s_KF}, where the minimum state estimation error is achieved when $s_{real}$ is known. Therefore, Figure \ref{fig:SSE_s_KF}, \ref{fig:s_est_time} and \ref{fig:F_peaks}, together demonstrates the validity of our observer design in simulation. In the next section, we will benchmark our observer against the state-of-the-art observers.

\begin{figure}[h]
    \centering
    \captionsetup{justification=justified}
    \includegraphics[scale = 0.3]{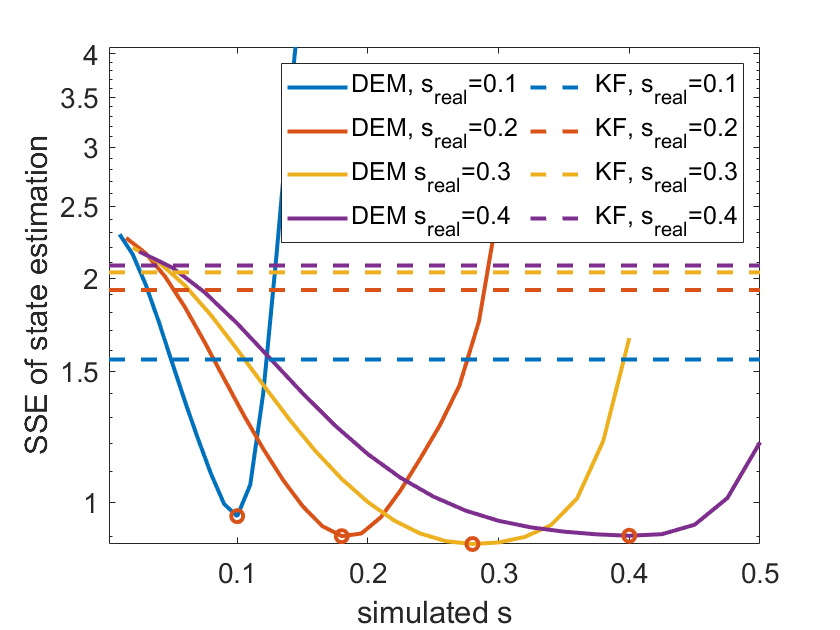}
    \caption{The sum of squared error (SSE) in state estimation of DEM deteriorates when the noise smoothness $s$ used is different from the real noise smoothness $s_{real}$. The solid and dotted lines denote the SSE of DEM and KF for different simulated $s$. The SSE for DEM takes a minimum when $s \approx s_{real}$. For lower $s$ (0.1 for example), KF outperforms DEM if $s$ is not close to $s_{real}$, emphasizing the importance of an online noise smoothness observer.}
    \label{fig:SSE_s_KF}
\end{figure}
 
\section{BENCHMARKING} \label{sec:benchmarking}
This section aims to benchmark the performance of our smoothness estimator for a state estimation problem. Through rigorous simulations, we show that our observer provides competitive performance during high colored noise cases.

% \subsection{Metric for comparison}
% We use Sum of Squared Error (SSE) in state estimation for the full time sequence as the metric for benchmarking:
% \begin{equation}
%     SSE = \sum_{t=0}^{T} (\textbf{x}(t) - x(t))^2.
% \end{equation}
% Lower the estimation error, higher the estimation accuracy of the observer.

\subsection{Embedding order of states}
In this section, we use rigorous simulations to show that our observer design can enable state estimation under a wide range of noises -- at different embedding orders and smoothness levels. We manipulate on the dimension and component values of the $S$ matrix in Equation \ref{eqn:S_matrix} through different $p$ and $s$ values, under the same simulation setup described in Section \ref{sec:sim_setup} with $dt=0.05s$. The size of $S$ matrix increases with increasing $p$, whereas the components inside it increases with increasing $s$. Figure \ref{fig:SSE_p} shows the results of state estimation using 150 experiments (5 randomly generated noises each for five $s$ values and six $p$ values). The estimation error decreases with increasing $p$ for different noise smoothness values, highlighting the importance of using higher order generalized coordination during estimation. This shows the applicability of our observer for a wide range of noise smoothness, embedding orders and noises.

\begin{figure}[h]
    \centering
    \captionsetup{justification=justified}
    \includegraphics[scale = .27]{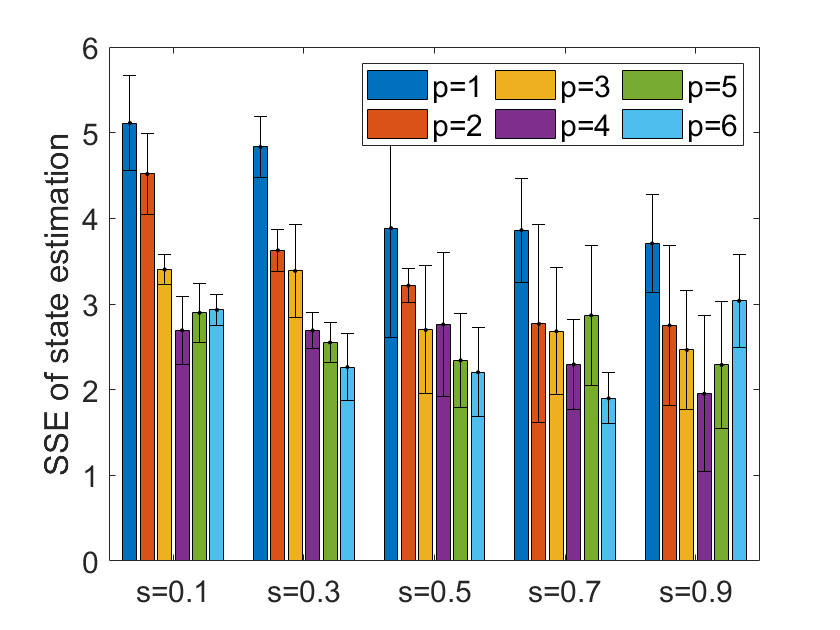}
    \caption{The error in state estimation decreases as the embedding order of states $p$ increases, for a range of noise smoothness $s$. This shows that our smoothness estimation aids an accurate state estimation till an embedding order of $p=4$ for a wide range of $s$.   }
    \label{fig:SSE_p}
\end{figure}

\subsection{Benchmark state observer}
In this section, we benchmark our observer against other state-of-the-art observers like KF, SA and SMIKF, to show its competitiveness. 50 time series data (10 each for 5 smoothness values with $dt=0.05s$) were generated using the simulation setup in Section \ref{sec:sim_setup} and the SSE in state estimation was computed for KF, SA, SMIFK and DEMs. The SMIKF and SA implementation accommodated an AR model of order 1 and 6 respectively for the noise modelling, whereas the DEM implementation used an embedding order of $p=6$ for states and $d=2$ for inputs. Figure \ref{fig:benchmark_DEM_SSE} shows the results, clearly indicating the superior performance of DEMs with minimum error in state estimation for higher $s$. DEMs outperforms other observers for a wide range of $s$ values. However, for low noise color ($s=0.1$), SA and SMIKF outperforms DEMs. In all cases, DEMs outperforms KF in the presence of colored noise.

% For different noise smoothness, noise variance and with other observers.

% \subsection{Fully randomized benchmarking}
% For different system orders and parameter values.

\begin{figure}[h]
    \centering
    \captionsetup{justification=justified}
    \includegraphics[scale = 0.26]{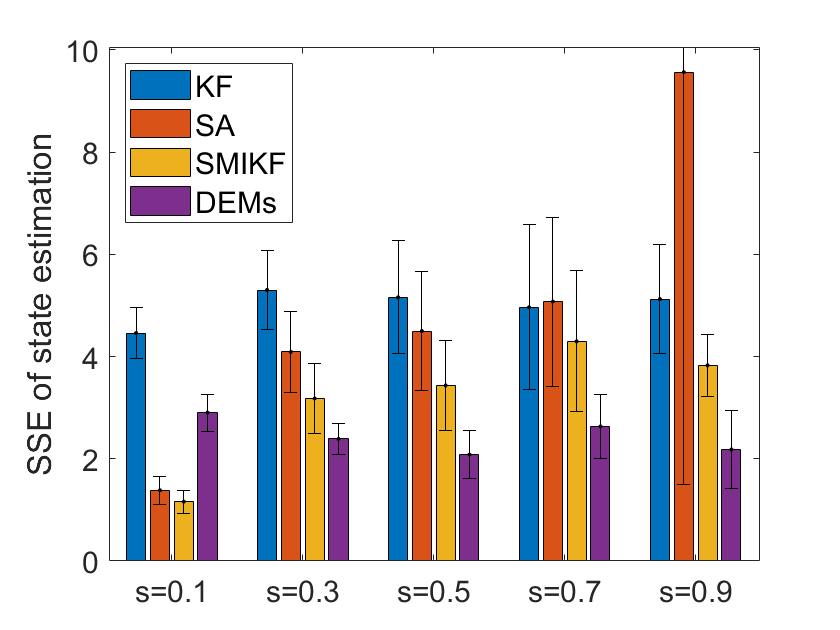}
    \caption{DEMs outperform KF, SA and SMIKF with minimal estimation error during state estimation under high colored noise ($s>0.1$). When noises are near white ($s$ close to $dt$), DEMs outperforms KF, but SA and SMIKF performs better. Solutions of SA is unstable for higher $s$. }
    \label{fig:benchmark_DEM_SSE}
\end{figure} 

\section{PROOF OF CONCEPT - QUADROTOR FLIGHT} \label{sec:quadrotor}
This section aims to provide a proof of concept for our observer design by employing it for the state estimation of a real quadrotor flying under wind conditions. We use the experimental design from \cite{bos2021free} to obtain the quadrotor flight data. The experiment consist of a quadrotor hovering at a fixed location, under the strong influence of wind generated by a blower. The linearized quadrotor model relating the input motor signals to the output roll angle ($\phi$) of the quadcopter, without accounting for the wind dynamics is given by \cite{bos2021free}:

\begin{equation}
\label{eqn:ss_small}
\begin{aligned}
    \begin{bmatrix}
    \dot \phi\\
    \ddot \phi
    \end{bmatrix} &= \begin{bmatrix}
    0&1\\0&0
    \end{bmatrix}\begin{bmatrix}
    \phi \\ \dot \phi
    \end{bmatrix} + \begin{bmatrix}
        0&0&0&0\\
        \frac{c_{B\phi}}{I_{xx}}&-\frac{c_{B\phi}}{I_{xx}}&-\frac{c_{B\phi}}{I_{xx}}&\frac{c_{B\phi}}{I_{xx}}
    \end{bmatrix} \Bigg[ \begin{smallmatrix}
    pwm_1\\pwm_2\\pwm_3\\pwm_4
    \end{smallmatrix} \Bigg], \\
    y &= \begin{bmatrix}
    1&0
    \end{bmatrix}\begin{bmatrix}
    \phi\\
    \dot \phi
    \end{bmatrix},
\end{aligned}
\end{equation}
where $pwm_i$ is the Pulse Width Modulation signal provided to the $i^{\text{th}}$ motor by the controller for stable hovering, $I_{xx} = 3.4 \cdot 10^{-3} kgm^2$ is the quadcopter's moment of inertia around the $x$-axis, and  $c_{B\phi} = 1.274 \cdot 10^{-3}Nm$ is the thrust coefficient that models the relation between the PWM values and the thrust generated by the quadcopter rotors. $\phi$ was recorded using the Optitrack system, and was used for the state estimation for a time sequence of $T = 15s$ with $dt = 0.0083s$. The influence of wind dynamics on the quadrotor states ($\phi$ and $\dot{\phi}$) is unmodelled in Equation \ref{eqn:ss_small}. Therefore, the wind dynamics induces strong colored noise ($w$) in the data \cite{meera2021brain}. The higher process noise ($\Pi^w = e^4$), and a lower measurement noise ($\Pi^z = e^{10}$) were used to represent high unmodelled wind noise and low Optitrack noise respectively. The embedding order of $p = 2$ and $d= 2$ were used to capture the noise color.

\begin{figure}[h]
    \centering
    \captionsetup{justification=justified}
    \includegraphics[scale = 0.27]{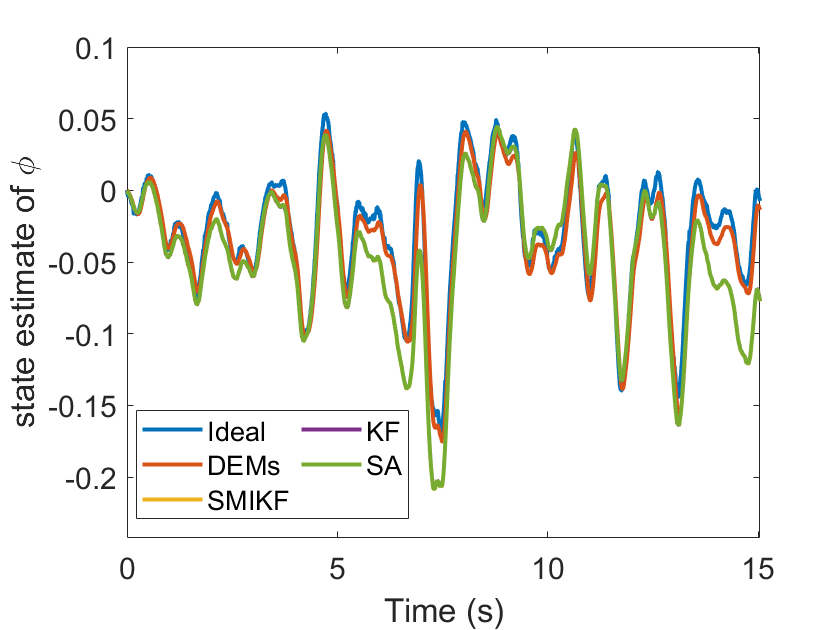}
    \caption{DEMs outperforms other benchmarks in state estimation on the quadrotor flight data where it hovers under the influence of wind, introducing colored noise into the system. DEMs (in red) is closer to the ground truth Optitrack measurement (ideal in blue), when compared to other benchmarks. KF, SA and SMIKF shows coinciding estimation plots.}
    \label{fig:drone_state_est_bench}
\end{figure}

\begin{figure}[h]
    \centering
    \captionsetup{justification=justified}
    \includegraphics[scale = 0.23]{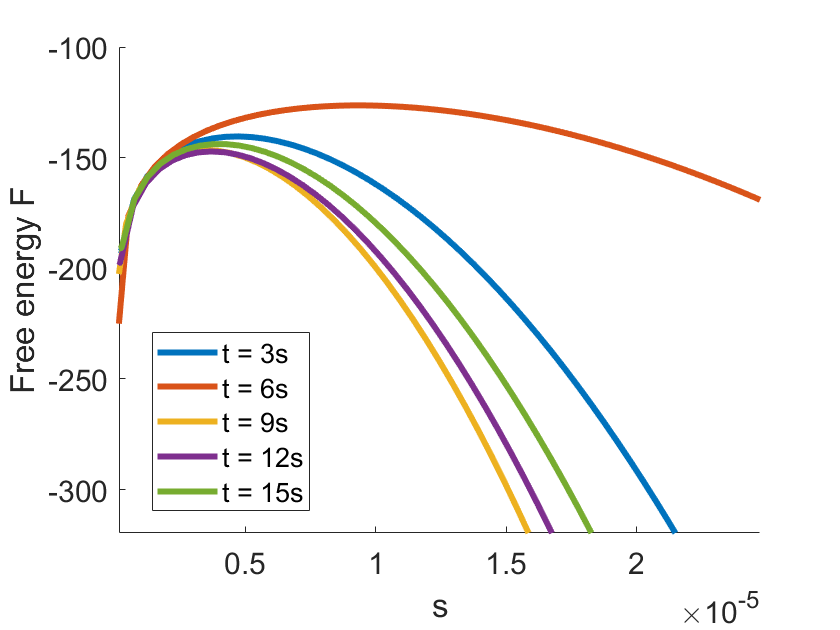}
    \caption{The free energy vs $s$ plot for different time instances during the quadrotor flight (data from Figure \ref{fig:drone_state_est_bench}). The curves show a clear maximum, similar to the simulation results from Figure \ref{fig:F_peaks}. This provides an experimental validation for using the gradient ascend over free energy for smoothness estimation.  }
    \label{fig:drone_s_est}
\end{figure} 

Figure \ref{fig:drone_state_est_bench} shows the superior state estimation capabilities of DEMs. DEMs (in red) is closer to the ground truth (in blue) when compared to other benchmarks. KF, SMIKF and SA have coinciding state estimation curves. Figure \ref{fig:drone_s_est} shows the free energy vs $s$ curve at different time instances $t$ of the quadrotor flight, showing a clear maximum, similar to the simulation results in Figure \ref{fig:F_peaks}, validating the practical application of our estimator.

\section{CONCLUSION}
A novel observer (DEMs) for the joint state and noise smoothness estimation of linear systems with colored noise was introduced. Through rigorous simulations, DEMs was shown to outperform the benchmarks like KF, SMIKF and SA in state estimation under colored noise with minimum estimation error. The observer was face validated by applying it on a practical robotics application - the state estimation of a quadrotor hovering in unmodelled wind conditions, to show that DEMs is a competitive observer. The main limitation of this work is the absence of a stability proof for the joint state and noise smoothness observer, which can be the focus of future research. The estimator could be extended to solve the general active inference problem for the estimation and control of nonlinear systems with colored noise.

\section*{Acknowledgment}
We would like to thank Peyman Mohajerin Esfahani for his valuable insights on joint observer design.

\bibliographystyle{IEEEtran}
\footnotesize
\bibliography{IEEEabrv,reference}

% Generated by IEEEtran.bst, version: 1.14 (2015/08/26)
\begin{thebibliography}{10}
\providecommand{\url}[1]{#1}
\csname url@samestyle\endcsname
\providecommand{\newblock}{\relax}
\providecommand{\bibinfo}[2]{#2}
\providecommand{\BIBentrySTDinterwordspacing}{\spaceskip=0pt\relax}
\providecommand{\BIBentryALTinterwordstretchfactor}{4}
\providecommand{\BIBentryALTinterwordspacing}{\spaceskip=\fontdimen2\font plus
\BIBentryALTinterwordstretchfactor\fontdimen3\font minus
  \fontdimen4\font\relax}
\providecommand{\BIBforeignlanguage}[2]{{%
\expandafter\ifx\csname l@#1\endcsname\relax
\typeout{** WARNING: IEEEtran.bst: No hyphenation pattern has been}%
\typeout{** loaded for the language `#1'. Using the pattern for}%
\typeout{** the default language instead.}%
\else
\language=\csname l@#1\endcsname
\fi
#2}}
\providecommand{\BIBdecl}{\relax}
\BIBdecl

\bibitem{crassidis2004optimal}
J.~L. Crassidis and J.~L. Junkins, \emph{Optimal estimation of dynamic
  systems}.\hskip 1em plus 0.5em minus 0.4em\relax Chapman and Hall/CRC, 2004.

\bibitem{zhouSMIKF}
Z.~Zhou, J.~Wu, Y.~Li, C.~Fu, and H.~Fourati, ``Critical issues on kalman
  filter with colored and correlated system noises,'' \emph{Asian Journal of
  Control}, vol.~19, no.~6, pp. 1905--1919, 2017.

\bibitem{brysonSA}
A.~Bryson and D.~Johansen, ``Linear filtering for time-varying systems using
  measurements containing colored noise,'' \emph{IEEE Transactions on Automatic
  Control}, vol.~10, no.~1, pp. 4--10, 1965.

\bibitem{brysonMD}
A.~Bryson~Jr and L.~Henrikson, ``Estimation using sampled data containing
  sequentially correlated noise.'' \emph{Journal of Spacecraft and Rockets},
  vol.~5, no.~6, pp. 662--665, 1968.

\bibitem{friston2008variational}
K.~J. Friston, N.~Trujillo-Barreto, and J.~Daunizeau, ``Dem: a variational
  treatment of dynamic systems,'' \emph{Neuroimage}, vol.~41, no.~3, pp.
  849--885, 2008.

\bibitem{friston2010free}
K.~Friston, ``The free-energy principle: a unified brain theory?'' \emph{Nature
  reviews neuroscience}, vol.~11, no.~2, pp. 127--138, 2010.

\bibitem{meera2020free}
A.~A. Meera and M.~Wisse, ``Free energy principle based state and input
  observer design for linear systems with colored noise,'' in \emph{2020
  American Control Conference (ACC)}.\hskip 1em plus 0.5em minus 0.4em\relax
  IEEE, 2020, pp. 5052--5058.

\bibitem{bos2021free}
F.~Bos, A.~A. Meera, D.~Benders, and M.~Wisse, ``Free energy principle for
  state and input estimation of a quadcopter flying in wind,'' \emph{arXiv
  preprint arXiv:2109.12052}, 2021.

\bibitem{friston2009free}
K.~Friston, ``The free-energy principle: a rough guide to the brain?''
  \emph{Trends in cognitive sciences}, vol.~13, no.~7, pp. 293--301, 2009.

\bibitem{friston2010action}
K.~J. Friston, J.~Daunizeau, J.~Kilner, and S.~J. Kiebel, ``Action and
  behavior: a free-energy formulation,'' \emph{Biological cybernetics}, vol.
  102, no.~3, pp. 227--260, 2010.

\bibitem{carhart2010default}
R.~L. Carhart-Harris and K.~J. Friston, ``The default-mode, ego-functions and
  free-energy: a neurobiological account of freudian ideas,'' \emph{Brain},
  vol. 133, no.~4, pp. 1265--1283, 2010.

\bibitem{friston2010generalised}
K.~Friston, K.~Stephan, B.~Li, and J.~Daunizeau, ``Generalised filtering,''
  \emph{Mathematical Problems in Engineering}, vol. 2010, 2010.

\bibitem{li2011generalised}
B.~Li, J.~Daunizeau, K.~E. Stephan, W.~Penny, D.~Hu, and K.~Friston,
  ``Generalised filtering and stochastic dcm for fmri,'' \emph{neuroimage},
  vol.~58, no.~2, pp. 442--457, 2011.

\bibitem{xiong2015optimal}
J.-J. Xiong and E.-H. Zheng, ``Optimal kalman filter for state estimation of a
  quadrotor uav,'' \emph{Optik}, vol. 126, no.~21, pp. 2862--2868, 2015.

\bibitem{lanillos2021active}
P.~Lanillos, C.~Meo, C.~Pezzato, A.~A. Meera, M.~Baioumy, W.~Ohata,
  A.~Tschantz, B.~Millidge, M.~Wisse, C.~L. Buckley \emph{et~al.}, ``Active
  inference in robotics and artificial agents: Survey and challenges,''
  \emph{arXiv preprint arXiv:2112.01871}, 2021.

\bibitem{oliver2019active}
G.~Oliver, P.~Lanillos, and G.~Cheng, ``Active inference body perception and
  action for humanoid robots,'' \emph{arXiv preprint arXiv:1906.03022}, 2019.

\bibitem{baioumy2021active}
M.~Baioumy, P.~Duckworth, B.~Lacerda, and N.~Hawes, ``Active inference for
  integrated state-estimation, control, and learning,'' in \emph{2021 IEEE
  International Conference on Robotics and Automation (ICRA)}.\hskip 1em plus
  0.5em minus 0.4em\relax IEEE, 2021, pp. 4665--4671.

\bibitem{meera2021brain}
A.~A. Meera and M.~Wisse, ``A brain inspired learning algorithm for the
  perception of a quadrotor in wind,'' \emph{arXiv preprint arXiv:2109.11971},
  2021.

\bibitem{ccatal2021robot}
O.~{\c{C}}atal, T.~Verbelen, T.~Van~de Maele, B.~Dhoedt, and A.~Safron, ``Robot
  navigation as hierarchical active inference,'' \emph{Neural Networks}, vol.
  142, pp. 192--204, 2021.

\bibitem{baltieri2019pid}
M.~Baltieri and C.~L. Buckley, ``Pid control as a process of active inference
  with linear generative models,'' \emph{Entropy}, vol.~21, no.~3, p. 257,
  2019.

\bibitem{baltieri2021kalman}
M.~Baltieri and T.~Isomura, ``Kalman filters as the steady-state solution of
  gradient descent on variational free energy,'' \emph{arXiv preprint
  arXiv:2111.10530}, 2021.

\bibitem{buckley2017free}
C.~L. Buckley, C.~S. Kim, S.~McGregor, and A.~K. Seth, ``The free energy
  principle for action and perception: A mathematical review,'' \emph{Journal
  of Mathematical Psychology}, vol.~81, pp. 55--79, 2017.

\bibitem{friston2007variational}
K.~Friston, J.~Mattout, N.~Trujillo-Barreto, J.~Ashburner, and W.~Penny,
  ``Variational free energy and the laplace approximation,'' \emph{Neuroimage},
  vol.~34, no.~1, pp. 220--234, 2007.

\bibitem{anilmeera2021DEM_LTI}
A.~Anil~Meera and M.~Wisse, ``Dynamic expectation maximization algorithm for
  estimation of linear systems with colored noise,'' \emph{Entropy}, vol.~23,
  no.~10, p. 1306, 2021.

\end{thebibliography}

\appendix
\subsection{Gradients of \texorpdfstring{$\ln |\tilde{\Pi}|$}{Lg}} \label{app:gradients}
The log determinant of generalized precision can be calculated using Equations \ref{eqn:S_matrix} and \ref{eqn:gen_precision}  as:
\begin{equation}
\begin{split}
    \ln |\tilde{\Pi}| = & \ln |S \otimes \Pi^z|  + \ln |S \otimes \Pi^w|  \\
     = &  \ln (|S|^{m} |\Pi^z|^{(p+1)})  + \ln (|S|^{n} |\Pi^w|^{(p+1)}) \\
     =  &  (p+1) ( \ln |\Pi^z| +   \ln |\Pi^w| ) + (n+m) \ln |S| .
\end{split}
\end{equation}
The derivative of $\ln |\tilde{\Pi}|$ with respect to $s$ becomes:
\begin{equation} \label{eqn:dlogPi_ds}
        \frac{\partial \ln |\tilde{\Pi}|}{\partial s}  = (n+m) \frac{\partial \ln |S|}{\partial s}, 
        \frac{\partial^2 \ln |\tilde{\Pi}|}{\partial s^2}  = (n+m) \frac{\partial^2 \ln |S|}{\partial s^2}.
\end{equation}
From Equation \ref{eqn:S_matrix}, $|S| = \frac{512}{6075}s^{42}$, resulting in $\frac{\partial \ln |S|}{\partial s} = \frac{42}{s}$ and $\frac{\partial^2 \ln |S|}{\partial s^2} = -\frac{42}{s^2}$. This simplifies Equation \ref{eqn:dlogPi_ds} to:
\begin{equation} 
        \frac{\partial \ln |\tilde{\Pi}|}{\partial s}  = 42(n+m) \frac{1}{s}, \
        \frac{\partial^2 \ln |\tilde{\Pi}|}{\partial s^2}  = -42(n+m) \frac{1}{s^2}. 
\end{equation}

\subsection{Numerical analysis on the nature of \texorpdfstring{$\tilde{\epsilon}^T \tilde{\Pi} \tilde{\epsilon}$}{Lg}} \label{app:numerical}
We recorded the first two gradients of the polynomial $\tilde{\epsilon}^T \tilde{\Pi} \tilde{\epsilon}$ with respect to $s$ for 20,000 combinations of randomly sampled $\tilde{\epsilon}$ and $s$ such that $|\tilde{\epsilon}| <1 $ and $s \in (0,1]$. From the results shown in Figure \ref{fig:monotonic}, the data points predominantly lie on the first quadrant, suggesting that the function has positive gradients, which is a sign of monotonically increasing function. The absence of any points on the fourth quadrant motivates the conclusion: if $\tilde{\epsilon}^T \tilde{\Pi}_s \tilde{\epsilon} > 0$ then $\tilde{\epsilon}^T \tilde{\Pi}_{ss} \tilde{\epsilon} > 0$. The results remain the same for different norm lengths of $\tilde{\epsilon}$.

\begin{figure}[h]
    \centering
    \captionsetup{justification=justified}
    \includegraphics[scale = 1]{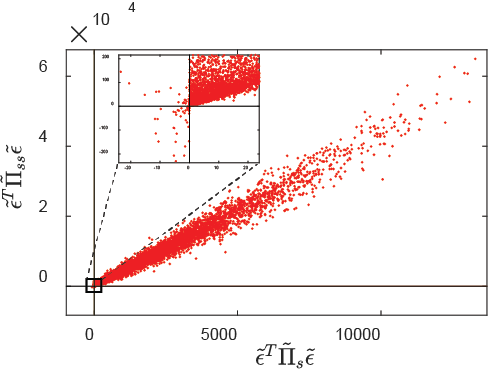}
    \caption{The plot demonstrating that the function $\tilde{\epsilon}^T \tilde{\Pi} \tilde{\epsilon}$ is mostly monotonically increasing with respect to $s$ in the domain (0,1] for $|\tilde{\epsilon}|<1$. Moreover, when $\tilde{\epsilon}^T \tilde{\Pi}_s \tilde{\epsilon} > 0$, $\tilde{\epsilon}^T \tilde{\Pi}_{ss} \tilde{\epsilon}< 0$, since there are no data points on the fourth quadrant as shown in the zoomed plot.  }
    \label{fig:monotonic}
\end{figure}

\addtolength{\textheight}{-3cm} 
% \vspace{12pt}

\end{document}